\documentclass[11pt,catalan,nofootinbib,preprintnumbers,longbibliography]{article}
\usepackage{fullpage} 
\usepackage[latin1]{inputenc}
\usepackage{amsmath,latexsym}
\usepackage{amsthm}
\usepackage{amssymb}
\usepackage{amsfonts}
\usepackage{appendix}
\usepackage[english]{babel}
\usepackage{times}
\usepackage{cite}
\usepackage{authblk}
\usepackage{xcolor}
\usepackage{tcolorbox}
\usepackage{pagecolor}
\usepackage[colorlinks=true,
            citecolor=red,
            linkcolor=blue,
            urlcolor=violet,
            filecolor=cyan,
            backref=false]{hyperref}

\usepackage{titlesec}

\DeclareMathOperator{\Erf}{Erf}

\newcommand{\D}{{\rm d}}
\newcommand{\calD}{\mathcal{D}}
\newcommand{\calK}{\mathcal{K}}
\newcommand{\calE}{\mathcal{E}}
\newcommand{\calS}{\mathcal{S}}

\newcommand{\calC}{\mathcal{C}}
\newcommand{\calJ}{\mathcal{J}}
\newcommand{\calF}{\mathcal{F}}
\newcommand{\calT}{\mathcal{T}}
\newcommand{\calTE}{\mathcal{TE}}
\newcommand{\RR}{\ensuremath{ \mathbb{R}} }

\allowdisplaybreaks

\titleformat{\paragraph}
{\normalfont\normalsize\bfseries}{\theparagraph}{1em}{}
\titlespacing*{\paragraph}
{0pt}{3.25ex plus 1ex minus .2ex}{1.5ex plus .2ex}

\theoremstyle{definition}

\newtheorem{proposition}{Proposition}[section]
\newtheorem{definition}{Definition}[section]
\newtheorem{remark}{Remark}[section]

\title{Infinite-derivative linearized gravity in convolutional form}

\author[1]{\small Carlos Heredia\thanks{carlosherediapimienta@gmail.com}}
\author[2]{Ivan Kol\'a\v{r}\thanks{i.kolar@rug.nl}}
\author[1]{Josep Llosa\thanks{pitu.llosa@ub.edu}}
\author[3,4]{Francisco Jos\'e Maldonado Torralba\thanks{fmaldo01@ucm.es}}
\author[2]{Anupam Mazumdar\thanks{anupam.mazumdar@rug.nl}}

\affil[1]{\small \textit{Facultat de F\'{\i}sica (FQA and ICC)\\ Universitat de Barcelona, Diagonal 645, 08028 Barcelona, Catalonia, Spain}}
\affil[2]{\textit{Van Swinderen Institute, University of Groningen, 9747 AG, Groningen, Netherlands}}
\affil[3]{\textit{Cosmology and Gravity Group, Department of Mathematics and Applied Mathematics, University of Cape Town, Rondebosch 7701, Cape Town, South Africa}}
\affil[4]{\textit{Institute of Theoretical Astrophysics, University of Oslo, P.O. Box 1029 Blindern, N-0315 Oslo, Norway}}

\begin{document}

\maketitle

\begin{abstract}
This article aims to transform the infinite-order Lagrangian density for ghost-free infinite-derivative linearized gravity into non-local. To achieve it, we use the theory of generalized functions and the Fourier transform in the space of tempered distributions $\calS^\prime$. We show that the non-local operator domain is not defined on the whole functional space but on a subset of it. Moreover, we prove that these functions and their derivatives are bounded in all $\RR^3$ and, consequently, the Riemann tensor is regular and the scalar curvature invariants do not present any spacetime singularity. Finally, we explore what conditions we need to satisfy so that the solutions of the linearized equations of motion exist in $\calS^\prime$.
\end{abstract}

\section{Introduction}
\label{sec:Introduction}

Understanding gravity is one of the significant challenges of this century. Einstein's theory of gravity pioneered giving answers and generating unexpected results, such as black holes and gravitational waves. Unfortunately, there are still unsolved problems making us think that General Relativity (GR) is only an effective theory, which works exceptionally well at low energies but breaks down at the high energies (UV regime). Inspired by quantum field theory \cite{Efimov1967,Krasnikov1987, Moffat1990, Evens1991} and quantum gravity \cite{Tomboulis1980,Tomboulis1997,Modesto2012} (for instance, string theory \cite{WITTEN1986, Polchinski1998} among others \cite{Ashtekar2013}), the behaviour of the theory in the UV regime was improved by the presence of \textit{non-locality} in the Lagrangian.

\vspace{0.1cm}
\textit{Infinite Derivative Gravity (IDG)} is a modified theory of gravity that generally behaves better in the UV regime and gives a chance to avoid cosmological and black hole singularities \cite{Biswas2012,Biswas2010}. IDG is made up of form factors, which encapsulate infinite derivatives in series form. More precisely, these form factors are analytic functions with no roots in the complex plane that guarantee the preservation of the same degrees of freedom as in Einstein's gravity, when it is perturbed around particular backgrounds. It is achieved by following a specific criterion \cite{Biswas2013}. This particular choice within IDG is known as \textit{ghost-free infinite-derivative gravity} and is characterised by a mass scale $M_s$ at which non-local effects become relevant. 

\vspace{0.1cm}
There are many examples where IDG shows that regularisation of the gravitational field is possible, for instance, the well-known gravitational potential $1/r$ of pointlike sources at the linearized level \cite{Biswas2012}. A similar property remains true for other types of sources, for example: electromagnetic and NUT charges \cite{Buoninfante2018_2, Kol2020, Boos2021}, accelerated particles \cite{Kol2021}, models of mini-black-hole production \cite{Frolov2015_2}, scalar lumps \cite{Buoninfante2021}, spinning ring distributions \cite{Buoninfante2018_3} and other objects associated with topological defects such as p-branes, cosmic strings and gyratons \cite{Boos2018, Boos2020_3, Boos2020_2}. Furthermore, it was shown that IDG also provides solutions for bouncing cosmology \cite{Biswas2006,Koshelev2012,Biswas2012_2,Kolar2021,Kumar2020,Kumar2021} and gravitational waves \cite{Kilicarslan2019,Dengiz2020,Kol2021_2,Kolar2021_3}. 

\vspace{0.1cm}
Due to the presence of the infinite number of derivatives (or derivatives that appear non-polynomially) in the Lagrangian, the information provided to determine a solution to the field equations ---a priori--- should be infinite. Therefore, there is nothing similar to existence and uniqueness theorems under these circumstances to confirm that the solution exists and is unique. This problem is known as the \textit{initial value problem} or the \textit{Cauchy problem}. However, recent studies \cite{Barnaby2008, Gorka2012, CALCAGNI2008} show the initial value problem might be well-posed even involving infinite derivatives or integro-differential equations. 

\vspace{0.1cm}
Working with infinite series drives us to control their convergence when they act on the functions of class $\calC^\infty(\RR^4)$. Skipping this fact could lead us to the mistake of working with non-convergent series, which would be meaningless. Although such results could be mathematically significant as formal series, numerical results are expected to compare with experimental data.

\vspace{0.1cm}
To avoid this convergence problem, we replace the infinite-derivative linear operator with an integral operator, a convolution operator in the case of translation invariance.  To this end, the theory of generalized functions will be convenient, and the Fourier transform in the space of tempered distributions $\calS^\prime(\RR^4)$ will be used \cite{Vladimirov_GF, Vladimirov,Hormander1990}.
 
\vspace{0.1cm}
As shown below, the space of functions on which the infinite-derivative linear operator is well-defined will be given by a subset of it. Consequently, the functions and their derivatives belonging to this subset ---in the static case--- will be bounded by a polynomial; in other words, they will be regular in all $\RR^3$. Therefore, the Riemann tensor, which involves the second derivatives of the field, will be finite everywhere, and all the scalar curvature invariants shall not present any spacetime singularity.

\vspace{0.1cm}
This approach transforms the infinite-order Lagrangian into a non-local one in which the infinite series are avoided. We will study the solutions of this non-local Lagrangian and observe that they will not be able to be slow-growing functions. Due to the complexity of the initial value problem, we will analyze the homogeneous and static cases separately. We will adopt the point of view of \cite{Llosa1994, Heredia2, Gomis2001} where the equations of motion are taken as constraints that limit the functional space. Within this approach, we will explore the conditions that the Hilbert energy-momentum tensor of matter might satisfy to obtain a solution for both cases. We will show that the solution always exists for the static case if the Hilbert energy-momentum tensor in the Fourier space is built by bounded functions.  

\vspace{0.1cm}
The plan of the article is as follows: In Section (\ref{sectionII}), the infinite-derivative Lagrangian for IDG is introduced; in Section (\ref{sectionIII}), we set up the ghost-free non-local operator and, using the inverse Fourier transform in $\mathcal{S}^\prime(\RR^4)$, we convert it into a convolution operator; in Section (\ref{sectionIV}), we rewrite the infinite-order Lagrangian for IDG in a non-local one; in Section (\ref{sectionV}), the solution of the field equations and what conditions must be satisfied so that the solution exists are discussed. Appendix (\ref{A:S}) is devoted to reviewing some notions of tempered distributions, $\mathcal{S}^\prime(\RR^4)$, addressed to unfamiliar readers with this topic.

\section{Infinite-derivative linearized gravity}
\label{sectionII}

The most general infinite-order, parity-invariant and torsion-free action in 4 dimensions is \cite{Biswas2012}
\begin{equation}
\label{S}
S([g_{\alpha\beta}]) = \frac{1}{16\pi} \int_{\mathbb{R}^4} \D x \sqrt{-g}\left[R + \frac{1}{2} \left( RF_1\left(\Box_s\right)R + R_{\mu\nu}F_2\left(\Box_s\right) R^{\mu\nu} + R_{\mu\nu\sigma\lambda} F_3\left(\Box_s\right) R^{\mu\nu\sigma\lambda}\right)\right] 
\end{equation}
where $F_i(\Box_s)$ are analytic functions of d'Alembertian $\Box_s := \Box/M^2_s = \nabla_{\mu}\nabla^{\mu}/M^2_s$, which are called \textit{form factors}, and $[g_{\alpha\beta}]$ means the functional dependence of the metric field $g_{\alpha\beta}$. We define the integral action (\ref{S}) as a functional on the space of all possible fields whether or not they satisfy the field equations. We call this space as the \textit{kinematic space} $\mathcal{K}\in\mathcal{C}^\infty(\mathbb{R}^4)$. 

\vspace{0.1cm}
As far as this article is concerned, we focus on the first-order expansion of the action (\ref{S}) around the Minkowski background $\eta_{\alpha\beta}= \mathrm{diag}(-1,1,1,1)$ in the Cartesian coordinates $(t,x,y,z)$,
\begin{equation}
g_{\alpha\beta}(x) = \eta_{\alpha\beta} + h_{\alpha\beta}(x) \qquad \mathrm{where} \qquad \left|h_{\alpha\beta}(x)\right| \ll 1 \;.
\end{equation}
According to \cite{Biswas2012}, the infinite-derivative gravity Lagrangian density for a small perturbation $h_{\mu\nu}(x)$ around Minkowski spacetime is
\begin{equation}
\label{Lna}
\begin{split}
\mathcal{L}([h_{\alpha\beta}], x) = \frac{1}{2\kappa}&\left[\frac{1}{2} h^{\mu\nu}(x)a(\Box)\Box h_{\mu\nu}(x) - h^{\mu\nu}(x)a(\Box) \partial_\mu \partial_\alpha h^\alpha_\nu(x) + h^{\mu\nu}(x) c(\Box) \partial_\mu \partial_\nu h(x) \right.\\
&\left. - \frac{1}{2} h(x) c(\Box) \Box h(x) + \frac{1}{2} h^{\mu\nu}(x)\frac{a(\Box) - c(\Box)}{\Box}\partial_\mu \partial_\nu\partial_\alpha\partial_\beta h^{\alpha\beta}(x)\right] 
\end{split}
\end{equation}
where $\kappa = 8\pi G$ stands for Einstein's gravitational constant, $h(x)= \eta^{\alpha\beta} h_{\alpha\beta}(x)$, and $a(\Box)$ or $c(\Box)$ are entire functions with no zeros in the complex satisfying $a(0) = c(0) = 1$. These entire functions can be expanded as formal Taylor series
\begin{equation}
\label{Ts}
a(\Box) = \sum^{\infty}_{n=0} \frac{a^{(n)}(0)}{n!} \Box^n\:, 
\end{equation}
and similarly for $c(\Box)$. In what follows, we focus on the particular case in which non-local theories of linearized gravity are analytical $a(\Box) = c(\Box)$. Consequently, the Lagrangian density (\ref{Lna}) simplifies to
\begin{equation}
\label{La}
\begin{split}
\mathcal{L}([h_{\alpha\beta}],x) = \frac{1}{2\kappa}&\left[\frac{1}{2} h^{\mu\nu}(x)a(\Box)\Box h_{\mu\nu}(x) - h^{\mu\nu}(x)a(\Box) \partial_\mu \partial_\alpha h^\alpha_\nu(x)\right.\\
&\left. + h^{\mu\nu}(x) a(\Box) \partial_\mu \partial_\nu h(x) - \frac{1}{2} h(x) a(\Box) \Box h(x) \right]\:.
\end{split}
\end{equation}
From now on, we address the Lagrangian density as $\mathcal{L}([h_{\alpha\beta}],x):=\mathcal{L}(h_{\alpha\beta},x)$ where the functional dependency is understood although the square bracket does not emphasize it to make the notation simpler.

\vspace{0.1cm}
The integral action (\ref{S}) with the Lagrangian density (\ref{La}) might be divergent because the integration over the whole domain is unbounded. For this reason, we assume the variation of the integral action (\ref{S}) is finite for all variations of $\delta h_{\mu\nu}(x)$ with compact support. Then, the Euler-Lagrange equations are
\begin{equation}
\psi_{\mu\nu}(h_{\alpha\beta},x) = 0 \qquad \mathrm{with} \qquad \psi_{\mu\nu}(h_{\alpha\beta},x) := \int_{\mathbb{R}^4} \D y\:\lambda_{\mu\nu}(h_{\alpha\beta},y,x)
\end{equation}
where 
\begin{equation}
\label{lambda}
 \lambda_{\mu\nu}(h_{\alpha\beta},y,x) := \frac{\delta \mathcal{L}(h_{\alpha\beta},y)}{\delta h^{\mu\nu} (x)}\:.
\end{equation}
In our case, the Euler-Lagrange equations corresponding to (\ref{La}) become
\begin{equation}
\label{EOM1}
\begin{split}
&a(\Box)\left[\Box h_{\mu\nu}(x) - \partial_\sigma(\partial_\nu h_\mu^{\:\sigma}(x) + \partial_\mu h_\nu^{\:\sigma}(x))\right.\\
&\left.+\eta_{\mu\nu}\left(\partial_\rho\partial_\sigma h^{\rho\sigma}(x) - \Box h(x)\right) + \partial_\mu\partial_\nu h(x)\right] = -2\kappa\:T_{\mu\nu}(x) 
\end{split}
\end{equation}
where $T_{\mu\nu}$ is the Hilbert energy-momentum tensor related to some Lagrangian density of matter $\mathcal{L}_M$,
\begin{equation}
\label{T}
T_{\mu\nu} (x):=\int_{\mathbb{R}^4} \D y\:\frac{\delta \mathcal{L}_M(h_{\alpha\beta},y)}{\delta h^{\mu\nu}(x)}\:.
\end{equation}
Rather than working with the metric perturbation $h_{\mu\nu}(x)$, we use the trace-reversed perturbation $\hat h_{\mu\nu}(x) := h_{\mu\nu}(x) - \frac{1}{2}\eta_{\mu\nu} h(x)$. Noticing that $\hat h(x) = - h(x)$ and imposing the gauge condition\footnote{Note that equations of motion (\ref{EOM1}) are gauge invariant --- there are more ``unknowns" than independent equations.} $\partial_\mu \hat h^{\mu\nu}(x)=0$ then the field equations simplify greatly and take the form
\begin{equation} 
\label{EOM}
a(\Box) \Box  \hat h_{\mu \nu}(x) = - 2\kappa\:T_{\mu\nu}(x) \;.
\end{equation}
This equation contains derivatives of any order of $\hat h_{\mu\nu}(x)$, i.e, it is an infinite-order differential equation \cite{Carmichael1936}. Extrapolating what we know about the initial data --- that is, the initial data that determines the solution is the field and all its time derivatives up to order $n-1$---, we would conclude (by analogy) that we need an infinite number of initial data to determine the solution. On the other hand, due to (\ref{Ts}), the left-hand side of (\ref{EOM}) is an infinite series, whose convergence should be guaranteed so that it all makes sense. Because ensuring convergence is not an easy task, we will instead convert the infinite-derivative linear operator $a(\Box)$ into an equivalent\footnote{This equivalence depends on the domain of operators and is still under study \cite{Barnaby2008,Carlsson2014}.}, namely, a convolution, whose domain is a subspace of $\mathcal{K}$ to ensure the summability of the integral operator.

\section{Non-local operator definition}
\label{sectionIII}

Let us consider the linear operator of infinite order $\calD(\partial) = \sum^\infty_{|\alpha|=0} a_\alpha \partial_\alpha$ acting on functions $h_{\mu\nu}(x)\in\mathcal{C}^\infty(\RR^4)$, where $\alpha=(\alpha_1,...,\alpha_n)$\footnote{Along the article, $\alpha$ will also be used as spacetime tensor index. For that reason, we will explicitly state when referring to the multi-index notation. }, $|\alpha|:= n$, $a_\alpha$ are constant coefficients and $\alpha_j = 1,...,4$. Proceeding as follows,
\begin{eqnarray}   \label{e1a}
  \mathcal{D}(\partial) h_{\mu\nu}(x) &=& \sum_{|\alpha|=0}^\infty a_{\alpha}\,\partial_\alpha\,\left[\frac1{(2\pi)^{4}} \int_{\RR^4} \D k \,e^{ikx} \,\tilde h_{\mu\nu}(k) \right]\\[1ex]\label{e1b}
   &=& \frac1{(2\pi)^{4}}\,\sum_{|\alpha|=0}^\infty  \int_{\RR^4} \D k \,a_{\alpha}\,(i)^{|\alpha|}\, k_\alpha \,e^{ikx} \,\tilde h_{\mu\nu}(k) \\[1ex]\label{e1c}
   &=& \frac1{(2\pi)^{4}}\,\int_{\RR^4} \D k \,e^{ikx} \,\tilde h_{\mu\nu}(k)\,\sum_{|\alpha|=0}^\infty a_{\alpha}\,(i)^{|\alpha|}\, k_\alpha \,,
\end{eqnarray}
where $\tilde h_{\mu\nu}(k)$ is the Fourier transform\footnote{See definition (\ref{FT}) for the Fourier transform convention used.} of $h_{\mu\nu}(x)$ and $k_\alpha = k_{\alpha_1}\cdot...\cdot k_{\alpha_n}$, we get 
\begin{equation}
 \label{e2}
\mathcal{D}(\partial) h_{\mu\nu}(x) = \frac1{(2\pi)^4}\,\int_{\RR^4} \D k \,e^{ikx} \,\tilde h_{\mu\nu}(k)\,\tilde{\mathcal{D}}(i k)
\end{equation}
where $\tilde{\mathcal{D}}(i k):= \sum_{|\alpha|=0}^\infty a_{\alpha}\,(i)^{|\alpha|}\, k_\alpha\,$. Note that the validity of steps from (\ref{e1a}) to (\ref{e1c}) requires that
\begin{enumerate}
\item $h(x)$ has Fourier transform, and that the Fourier integral theorem holds,
\item the differentiation under the integral sign and,
\item the series and the integral commute.
\end{enumerate}
As $h_{\mu\nu}(x)$ are smooth ---i.e., they belong to class $\mathcal{C}^\infty(\RR^4)$---, it is sufficient that $h_{\mu\nu}(x)\in L(\RR^4)$ so that the Fourier integral theorem holds \cite{Apostol1976}.

\vspace{0.1cm}
\noindent
Requiring $h_{\mu\nu}(x)$ to be summable is a rather restrictive condition. However, the transition from (\ref{e1a}) to (\ref{e1c}) can be made by avoiding the summability condition in the mark of the tempered distributions $\mathcal{S}^\prime(\RR^4)$. 

\vspace{0.1cm}
\noindent
The field functions $h_{\mu\nu}(x)$ are smooth and, if we assume that they are slow-growing ---more than any polynomial---, we can treat them as regular tempered distributions, i.e. distributions which are made by locally integrable functions of slow growth acting on $\mathcal{S}(\RR^4)$. 

\vspace{0.1cm}
\noindent
It is worth noticing that, although for physical reasons\footnote{It seems reasonable to require that the solutions do not consist of singular distributions, such as the Dirac Delta.}  $h_{\mu\nu}$ might be a regular distribution, its Fourier transform $\tilde{h}_{\mu\nu}$ may be singular.  A rather illustrative example would be $f(x) = x\in\mathcal{S}^\prime(\RR)$. The Fourier transform in $\mathcal{S}^\prime(\RR)$ is $(\tilde f,\varphi) = -2\pi i (\delta,\varphi^\prime)$ with $\varphi\in\calS$, which is clearly a singular distribution due to the Dirac Delta.

\begin{proposition}
\label{def:nonoperator}
Let $h_{\mu\nu}$ be a regular tempered distribution and $\calD(\partial)$ be the non-local operator. The non-local operator acts on $h_{\mu\nu}$ as follows
\begin{equation}
\calD(\partial) h_{\mu\nu} = \mathcal{F}^{-1}\left[\tilde\calD(ik) \tilde h_{\mu\nu}\right]
\end{equation}
where $\tilde{\mathcal{D}}(i k):= \sum_{|\alpha|=0}^\infty a_{\alpha}\,(i)^{|\alpha|}\, k_\alpha$, $\tilde h_{\mu\nu} := \mathcal{F}\left[h_{\mu\nu}\right]$ and $\calF$ stands for the Fourier transform on $\calS^\prime(\RR^4)$.
\end{proposition}
\begin{proof}
Using the fact that the Fourier transform is a continuous linear operator on $\mathcal{S}^\prime$, we obtain
\begin{eqnarray}   \label{e3a}
  \mathcal{D}(\partial) h_{\mu\nu} &=& \sum_{|\alpha|=0}^\infty a_{\alpha}\,\partial_\alpha\,\mathcal{F}^{-1}\left[\tilde h_{\mu\nu}\right]\\[1ex]\label{e3b}
   &=&  \sum_{|\alpha|=0}^\infty  \mathcal{F}^{-1}\left[ \,a_{\alpha}\,(i)^{|\alpha|}\, k_\alpha \,\tilde h_{\mu\nu}\right] \\[1ex]\label{e3c}
   &=&  \mathcal{F}^{-1}\left[ \sum_{|\alpha|=0}^\infty  a_{\alpha}\,(i)^{|\alpha|}\, k_\alpha \,\tilde h_{\mu\nu}\right], \qquad \forall \varphi\in\calS(\RR^4)\:.
\end{eqnarray}
Now, defining 
\begin{equation}
\tilde{\mathcal{D}}(i k):= \sum_{|\alpha|=0}^\infty a_{\alpha}\,(i)^{|\alpha|}\, k_\alpha\:,
\end{equation}
which is a smooth function, we get
\begin{equation}
\calD(\partial) h_{\mu\nu} = \mathcal{F}^{-1}\left[\tilde\calD(ik) \tilde h_{\mu\nu}\right]\:.
\end{equation}
Notice that in this new framework, the step from (\ref{e3a}) to (\ref{e3b}) is a well-known property of Fourier transform --- see \cite{Vladimirov}, equation (9.1) ---  and the step from (\ref{e3b}) to (\ref{e3c}) is allowed because $\mathcal{F}$ and $\mathcal{F}^{-1}$ are a continuous linear operator on $\mathcal{S}^\prime$ --- see proposition (\ref{StoS}). 
\end{proof}

\subsection{Ghost-free non-local operator}

To avoid ghosts and retain the same number of degrees of freedom, one could choose any $a(\Box)$ entire function with no zeros. For this article, we will focus only on the class of theory in which the form factor is
\begin{equation}
\label{a}
a(\Box) := \exp[-\ell^2 \Box]
\end{equation}
where $\ell : = 1/M_s >0$ denotes the length scale at which non-local effects become important\footnote{There is an interesting line of research where even powers of the $\Box$ operator are considered, see for instance \cite{Kol2021}.}. This form factor is called ghost-free non-local operator and guarantees that we can always recover the local regime as $\ell\rightarrow0$ since $a(0) = 1$. It is worth pointing out the importance of the negative sign of $a(\Box)$. As shown in \cite{Biswas2012,Biswas2013}, this negative sign is crucial to recover the correct Newton potential. If $a(\Box)$ were with positive sign, it would lead to imaginary Newton potentials, which would indicate that the theory is not physical \cite{Talaganis2015}. 

\vspace{0.1cm}
Consider the non-local operator $\calD(\partial):= a(\Box)$, then
\begin{equation}
\tilde \calD(ik) = e^{\ell^2 k_\mu k^\mu} = e^{-\ell^2 \omega^2} e^{\ell^2 \mathbf{k}^2}\:.
\end{equation}
Due to the second factor, it is not guaranteed that $\tilde \calD(ik)\tilde h_{\mu\nu}$ is a tempered distribution. There is no problem with the first factor since multiplying a locally integrable function of slow growth ---if $\tilde h_{\mu\nu}$ is considered as a regular distribution--- by the fast damping smooth function $e^{-\ell^2 \omega^2}$ always yields a locally integrable slow-growing function; however, the latter is not true if the factor  $e^{\ell^2 \mathbf{k}^2}$ grows exponentially at infinity. A condition to guarantee that $\calD(ik) \tilde h_{\mu\nu}$ is a tempered distribution is that
\begin{equation}
\tilde\Phi_{\mu\nu}:= e^{\ell^2\mathbf{k}^2}\tilde h_{\mu\nu}\in\mathcal{S}^\prime(\RR^4)\:.
\end{equation}
The above condition is equivalent to saying:
\begin{equation}  \label{e6}
 \exists\,\tilde\Phi_{\mu\nu} \quad\mbox{tempered distribution, such that } \quad 
\tilde h_{\mu\nu} = e^{-\ell^2 \mathbf{k}^2}\,\tilde\Phi_{\mu\nu}\:.
\end{equation}
Furthermore, for the action integral to be meaningful, $e^{-\ell^2\Box}h_{\mu\nu}(x)$ must be a smooth function (of slow growth) --- or maybe a distribution whose domain includes smooth functions of slow growth as $h_{\mu\nu}(x)$ itself. Now, provided that (\ref{e6}) is fulfilled, 
\begin{equation}  \label{e7}
\tilde{\mathcal{D}}(i k)\:\tilde h_{\mu\nu} = e^{-\ell^2 \omega^2}\,\tilde\Phi_{\mu\nu} \quad \mbox{is a tempered distribution}\:.
\end{equation}

\begin{proposition}
\label{def:noperator}
Let $\Phi_{\mu\nu}$ be a regular tempered distribution, and
\begin{equation}
\label{kernels}
\calT_{(\ell)}(t) := \frac{1}{2\ell\sqrt{\pi}} e^{-\frac{t^2}{4\ell^2}} \qquad \mathrm{and} \qquad \calE_{(\ell)}(\mathbf{x}) := \frac{1}{(2\ell\sqrt{\pi})^3} e^{-\frac{|\mathbf{x}|^2}{4\ell^2}} 
\end{equation}
be smooth functions\footnote{In fact, they are the heat kernels in 1-dimension and 3-dimension space where $\ell^2$ plays the role of evolution parameter of the heat equation.}. The non-local operator $\calD(\partial) := a(\Box)$ acting on a regular tempered distribution $h_{\mu\nu}$ can be represented as a convolution 
\begin{equation}
a(\Box) h_{\mu\nu} = \left(\calT_{(\ell)}\ast\Phi_{\mu\nu}\right)
\end{equation} 
as long as 
\begin{equation}
\label{rel:hp}
h_{\mu\nu}  = \left(\calE_{(\ell)}\ast\Phi_{\mu\nu}\right)
\end{equation}
is satisfied. 
\end{proposition}

\begin{proof}
According to our assumption, condition (\ref{rel:hp}) must be satisfied so that we can represent the non-local operator in convolutional form. Therefore, applying the Fourier transform to (\ref{rel:hp}) in $\mathcal{S}^\prime$--- having in mind proposition (\ref{FourierConv}), remark (\ref{A2}) and the fact that $\Phi_{\mu\nu}\in\calS^\prime$ ---, we obtain 
\begin{equation}
\label{demrelhp}
\tilde h_{\mu\nu} = e^{-\ell^2 \mathbf{k}^2}\,\tilde\Phi_{\mu\nu}\:.
\end{equation}
Now, using proposition (\ref{def:nonoperator}), definition (\ref{a}) and equation (\ref{demrelhp}), we get
\begin{equation}
\label{proofconv1}
\left(a(\Box)h_{\mu\nu}, \varphi\right) = \left( \calF^{-1}\left[e^{-\ell^2\omega^2}\tilde\Phi_{\mu\nu}\right],\varphi\right) \qquad \varphi\in\calS(\RR^4)\:.
\end{equation}
By means of proposition (\ref{FourierConv}), the following relation in $\calS^\prime$ holds
\begin{equation}
\label{convrelationS}
f\ast g = \calF^{-1}\left[\calF\left[f\right]\cdot\calF\left[g\right]\right] \qquad \mathrm{with}\qquad  f\in\calS^\prime \quad g\in\calS\:.
\end{equation}
Consequently, because $e^{-\ell^2\omega^2}\in\calS(\RR)$ and $\tilde\Phi_{\mu\nu}\in\calS^\prime(\RR^4)$, we can relate them with (\ref{convrelationS}), namely,
\begin{equation}
\calF\left[f\right] := \tilde\Phi_{\mu\nu} \qquad \mathrm{and} \qquad \calF\left[g\right](\omega) : = e^{-\ell^2\omega^2}\:.
\end{equation}
Applying the inverse Fourier transform in $\calS^\prime$ to each of them, we obtain that 
\begin{equation}
\label{fg}
f := \Phi_{\mu\nu}  \qquad \mathrm{and} \qquad g (t) : = \calT_{(\ell)}  (t)
\end{equation}
where $\calT_{(\ell)}(t)$ is (\ref{kernels}). Finally, using (\ref{fg}) and (\ref{convrelationS}),  equation (\ref{proofconv1}) becomes
\begin{equation}
a(\Box)h_{\mu\nu} = \left(\calT_{(\ell)}\ast\Phi_{\mu\nu}\right) \quad \forall \varphi\:.
\end{equation}
\end{proof}

The $a(\Box) h_{\mu\nu}$ is correctly defined in $\mathcal{S}^\prime(\mathbb{R}^4)$ as long as condition (\ref{rel:hp}) is satisfied. On the other hand, $\Phi_{\mu\nu}$ can be treated as a locally integrable slow-growing function since, in proposition (\ref{def:noperator}), $\Phi_{\mu\nu}$ is considered as a regular tempered distribution\footnote{We would like to point out that proposition (\ref{def:noperator}) can be also extended to singular tempered distribution. However, we again impose that they are regular distributions since, in Section (\ref{sectionIV}), we will need to treat them as functions.}. As a consequence of this fact,  the kinematic space on which the non-local operator acts is reduced as follows
\begin{equation}
\label{kspace}
\calK =\left\{ h_{\mu\nu}(x)\in\calC^\infty(\RR^4) \quad | \quad h_{\mu\nu} (x) = \left(\calE_{(\ell)}\ast\Phi_{\mu\nu}\right)(x)  \right\}\:.
\end{equation}
\noindent
On the other hand, by using remark (\ref{A2}), we can assert that $\mathcal{E}_{(\ell)}(\mathbf{x})\in\mathcal{S}(\mathbb{R}^3)$. Therefore, the right-hand side of (\ref{rel:hp}) is the value of a regular distribution $\Phi_{\mu\nu}(t,\mathbf{x}^\prime)$ acting on the test function $\mathcal{E}_{(\ell)}(\mathbf{x}-\mathbf{x}^\prime)$. Consequently, the rigorous meaning of (\ref{rel:hp}) in the algebra $\mathcal{S}^\prime$ is \cite{Vladimirov}:
\begin{equation}
\label{eqhSp}
h_{\mu\nu}(t,\mathbf{x}) = \left(\Phi_{\mu\nu}(t,\mathbf{x}^\prime),\mathcal{E}_{(\ell)}(\mathbf{x}-\mathbf{x}^\prime)\right)\;.
\end{equation} 

\begin{proposition}
\label{bound}
Let $s\in\mathbb{N}$. For the static case, the tensor field $h_{\mu\nu}(\mathbf{x})$ is bounded by
\begin{equation}
|h_{\mu\nu}(\mathbf{x})| \leq C_{\mu\nu} (1 + |\mathbf{x}|^s)\;.
\end{equation}
where $C_{\mu\nu}$ is a constant tensor.
\end{proposition}

\begin{proof}
Let $D^\alpha$ be the derivative of order $\alpha$ in the multi-index notation and $P_{(\alpha)}(\mathbf{x})$ a polynomial of degree $\alpha$. By invoking Schwartz's theorem \cite{Vladimirov}, we get
\begin{align}
|h_{\mu\nu}(\mathbf{x})| &\leq C_{(1)\mu\nu} ||\mathcal{E}_{(\ell)}(\mathbf{x}-\mathbf{x}^\prime)||_s = C_{(1)\mu\nu} \underset{\mathbf{x}^\prime\in\mathbb{R}^3; |\alpha|\leq s}{\mathrm{sup}} (1+ |\mathbf{x}^\prime|^2)^{\frac{s}{2}} \left|D^\alpha_{\mathbf{x}^\prime} e^{-\frac{|\mathbf{x}-\mathbf{x}^\prime|^2}{4\ell^2}}\right|\nonumber\\
& \leq C_{(1)\mu\nu} \underset{\mathbf{x}^\prime\in\mathbb{R}^3; |\alpha|\leq s}{\mathrm{sup}} (1+ |\mathbf{x}^\prime|^2)^{\frac{s}{2}} \left|P_{(\alpha)}(\mathbf{x}-\mathbf{x}^\prime)\right| \left|e^{-\frac{|\mathbf{x}-\mathbf{x}^\prime|^2}{4\ell^2}}\right|\nonumber\\
& =  C_{(1)\mu\nu} \underset{\mathbf{x}^\prime\in\mathbb{R}^3; |\alpha|\leq s}{\mathrm{sup}} (1+ |\mathbf{x}^\prime + \mathbf{x} - \mathbf{x}|^2)^{\frac{s}{2}} \left|P_{(\alpha)}(\mathbf{x}-\mathbf{x}^\prime)\right| \left|e^{-\frac{|\mathbf{x}-\mathbf{x}^\prime|^2}{4\ell^2}}\right|\nonumber\\
& \leq C_{(1)\mu\nu} \underset{\mathbf{x}^\prime\in\mathbb{R}^3; |\alpha|\leq s}{\mathrm{sup}} (1+ |\mathbf{x} - \mathbf{x}^\prime|^2 + |\mathbf{x}|^2)^{\frac{s}{2}} \left|P_{(\alpha)}(\mathbf{x}-\mathbf{x}^\prime)\right| \left|e^{-\frac{|\mathbf{x}-\mathbf{x}^\prime|^2}{4\ell^2}}\right|\nonumber\\
&\leq C_{(2)\mu\nu} \underset{\mathbf{x}^\prime\in\mathbb{R}^3; |\alpha|\leq s}{\mathrm{sup}} \left[(1+ |\mathbf{x} - \mathbf{x}^\prime|^2)^{\frac{s}{2}} + |\mathbf{x}|^s\right] \left|P_{(\alpha)}(\mathbf{x}-\mathbf{x}^\prime)\right| \left|e^{-\frac{|\mathbf{x}-\mathbf{x}^\prime|^2}{4\ell^2}}\right|\:.
\end{align}
We can always use that
\begin{equation}
P_{(\alpha)}(\mathbf{x}-\mathbf{x}^\prime) \leq C_3 (1 + |\mathbf{x}-\mathbf{x}^\prime|^2)^{\frac{s}{2}} 
\end{equation}
with $C_{3} \in \mathbb{R}^{+}$ and $s/2 > \alpha$. Therefore, 
\begin{align}
|h_{\mu\nu}(\mathbf{x})| &\leq  C_{(4)\mu\nu} \underset{\mathbf{x}^\prime\in\mathbb{R}^3; |\alpha|\leq s}{\mathrm{sup}} \left[(1+ |\mathbf{x} - \mathbf{x}^\prime|^2)^{\frac{s}{2}} + |\mathbf{x}|^s\right](1 + |\mathbf{x}-\mathbf{x}^\prime|^2)^{\frac{s}{2}} \left|e^{-\frac{|\mathbf{x}-\mathbf{x}^\prime|^2}{4\ell^2}}\right|\nonumber\\
&\leq  C_{(4)\mu\nu}  \underset{\mathbf{x}^\prime\in\mathbb{R}^3; |\alpha|\leq s}{\mathrm{sup}} (1+ |\mathbf{x} - \mathbf{x}^\prime|^2)^s + |\mathbf{x}|^s(1 + |\mathbf{x}-\mathbf{x}^\prime|^2)^{\frac{s}{2}} \left|e^{-\frac{|\mathbf{x}-\mathbf{x}^\prime|^2}{4\ell^2}}\right|\nonumber\\
&\leq  C_{(4)\mu\nu}  \underset{\mathbf{x}^\prime\in\mathbb{R}^3; |\alpha|\leq s}{\mathrm{sup}} (1+ |\mathbf{x} - \mathbf{x}^\prime|^2)^s + |\mathbf{x}|^s(1 + |\mathbf{x}-\mathbf{x}^\prime|^2)^s \left|e^{-\frac{|\mathbf{x}-\mathbf{x}^\prime|^2}{4\ell^2}}\right| \nonumber\\
&\leq C_{(4)\mu\nu} (1+ |\mathbf{x}|^s) \underset{\mathbf{x}^\prime\in\mathbb{R}^3; |\alpha|\leq s}{\mathrm{sup}}(1 + |\mathbf{x}-\mathbf{x}^\prime|^2)^{s} \left|e^{-\frac{|\mathbf{x}-\mathbf{x}^\prime|^2}{4\ell^2}}\right|\nonumber\\
& \leq C_{(5)\mu\nu} (1+ |\mathbf{x}|^s)
\end{align}
where we have used the same procedure as remark (\ref{A2}) in the last step.
\end{proof}

\begin{proposition}
\label{bound2}
Let $s\in\mathbb{N}$ and $D^\alpha$ be the derivative of order $\alpha$ in the multi-index notation. For the static case, the derivatives of the tensor field $h_{\mu\nu}(\mathbf{x})$ are bounded by
\begin{equation}
|D^\alpha h_{\mu\nu}(\mathbf{x})| \leq C_{\mu\nu} (1 + |\mathbf{x}|^s)\;.
\end{equation}
where $C_{\mu\nu}$ is a constant tensor.
\end{proposition}
\begin{proof}
Invoking Schwartz's theorem \cite{Vladimirov}, we get
\begin{align}
|D^\alpha h_{\mu\nu}(\mathbf{x})| &\leq C_{(1)\mu\nu} ||D^\alpha_{\mathbf{x}} \mathcal{E}_{(\ell)}(\mathbf{x}-\mathbf{x}^\prime)||_s = C_{(1)\mu\nu} \underset{\mathbf{x}^\prime\in\mathbb{R}^3; |\beta|\leq s}{\mathrm{sup}} \left(1+ |\mathbf{x}^\prime|^2 \right)^\frac{s}{2}\left|D^\alpha_\mathbf{x}D^\beta_{\mathbf{x}^\prime} \calE_{(\ell)}(\mathbf{x}-\mathbf{x}^\prime) \right|\nonumber\\
& = C_{(1)\mu\nu} \underset{\mathbf{y}\in\mathbb{R}^3; |\beta|\leq s}{\mathrm{sup}} \left(1+ |\mathbf{x}- \mathbf{y}|^2 \right)^\frac{s}{2}\left|D^{\alpha+\beta} \calE_{(\ell)}(\mathbf{y}) \right|\nonumber\\
& \leq C_{(1)\mu\nu} \underset{\mathbf{y}\in\mathbb{R}^3; |\beta|\leq s}{\mathrm{sup}} \left(1+ |\mathbf{x}|^2 + |\mathbf{y}|^2 \right)^\frac{s}{2}\left|D^{\alpha+\beta} \calE_{(\ell)}(\mathbf{y}) \right|\nonumber\\
& \leq  C_{(1)\mu\nu} \underset{\mathbf{y}\in\mathbb{R}^3; |\beta|\leq s}{\mathrm{sup}} \left(1+ |\mathbf{x}|^2\right)^{\frac{s}{2}} \left(1+ |\mathbf{y}|^2 \right)^\frac{s}{2}\left|D^{\alpha+\beta} \calE_{(\ell)}(\mathbf{y}) \right|\nonumber\\
&\leq  C_{(1)\mu\nu} \left(1+ |\mathbf{x}|^2\right)^{\frac{s}{2}} \underset{\mathbf{y}\in\mathbb{R}^3; |\sigma|\leq s+|\alpha|}{\mathrm{sup}} (1+|\mathbf{y}|^2)^{\frac{s+|\alpha|}{2}} \left| D^\sigma\calE_{(\ell)}(\mathbf{y})\right|\nonumber\\
& \leq C_{(2)\mu\nu} \left(1+ |\mathbf{x}|^s\right)
\end{align}
where we have used the same procedure as remark (\ref{A2}) in the last step.
\end{proof}

Let us emphasize the relevance of the last two propositions. As we have just shown, the tensor field $h_{\mu\nu}(\mathbf{x})$ and its derivatives are bounded by a polynomial, i.e., they are regular in all $\RR^3$. This fact implies that, under the definition of the non-local operator via the inverse Fourier transform in the space of tempered distributions and the consequences it has ---equation (\ref{rel:hp})---, one can avoid the problem of spacetime singularities for the approximation where the field is static. As is well known, if the second derivatives of $h_{\mu\nu}(\mathbf{x})$ are regular, the Riemann tensor will be finite everywhere and, consequently, the scalar curvature invariants will be finite as well.

\vspace{0.1cm}
On the other hand, it is necessary to mention that the constraint on the kinematic space is given by the particular case in which the non-local operator is defined as (\ref{a}). It could be the case that, by choosing another definition for $a(\Box)$, the subset of functions for which this operator is well-defined changes. For instance, $a(\Box) := \mathrm{exp}(\ell^2\Box)$. For this particular case, note that the ``problematic" part would be given by the factor $e^{\ell^2\omega^2}$, since it would grow exponentially at infinity. Therefore, a tempered distribution $\tilde\Phi_{\mu\nu}$ would have to exist, such that $\tilde h_{\mu\nu} = e^{-\ell^2\omega^2}\tilde\Phi_{\mu\nu}$, so that $\tilde \calD(ik)\tilde h_{\mu\nu}$ is a tempered distribution. This suggests that the kinematic space might be different for each non-local operator\footnote{For the case where $a(\Box):= \sin(\Box)$ something similar happens. This non-local operator in Fourier space is $\tilde\calD(ik) = \cos(\mathbf{k}^2)\sin(\omega^2) - \sin(\mathbf{k}^2)\cos(\omega^2)$ which belongs $\calS^\prime$ as a regular distribution. Therefore, the factor $\tilde\calD(ik)\tilde  h_{\mu\nu}$ is itself already a tempered distribution without the need to introduce an auxiliary tempered distribution $\tilde\Phi_{\mu\nu}$. Thus, the kinematic space will be different from (\ref{kspace}). }.

\vspace{0.1cm}
In what follows, we will be more flexible in our notation when describing distributions. For example, we will denote the Dirac delta $\delta(x)$ instead of $(\delta,\varphi)$ with $\varphi\in\calS(\RR)$. The main reason for this choice is for ease of notation and readability in the following sections. However, recall that this notation describes the distribution acting on the variable $x$ of the test function $\varphi$.

\section{Non-local Lagrangian density}
\label{sectionIV}

We are in a position to transform the infinite-order Lagrangian (\ref{La}) into a non-local one. To do so, let us use proposition (\ref{def:noperator}) and condition (\ref{rel:hp}) but with the category of locally integrable functions of slow growth. After a bit of algebra, the Lagrangian density (\ref{La}) becomes
 \begin{equation}
 \label{LSp}
\mathcal{L}(\Phi_{\alpha\beta},x) = \frac{1}{4\kappa} M^{\mu\nu\alpha\beta\sigma\rho}\left(\calE_{(\ell)}\ast\Phi_{\mu \nu}\right)(x)\:\partial_\alpha\partial_\beta\left(\calT_{(\ell)}\ast\Phi_{\sigma\rho}\right)(x)
\end{equation}
where $x^\mu=(t,\mathbf{x})$,  and $M^{\mu\nu \alpha\beta\sigma f}$ is
\begin{equation}
M^{\mu\nu\alpha\beta\sigma\rho}:= \eta^{\mu \sigma}\eta^{\nu \rho}\eta^{\alpha\beta} - 2\eta^{\mu \alpha}\eta^{\beta \sigma}\eta^{\nu \rho} + 2\eta^{\mu \alpha}\eta^{\nu\beta}\eta^{\sigma\rho} - \eta^{\mu \nu}\eta^{\alpha\beta }\eta^{\sigma\rho}\:.
\end{equation}
Let us make three observations concerning this Lagrangian density.  Observe that the kinematic space $\calK$ is now coordinated in terms of $\Phi_{\mu\nu}(x)$, which do not necessarily belong to $C^\infty(\mathbb{R}^4)$. An illustrative example of the latter would be $f(\mathbf{x}) = \frac{1}{|\mathbf{x}|^2}$. This function is not infinitely differentiable at $|\mathbf{x}| = 0$, however, it is a locally integrable function of slow growth. Consequently, the space of functions considered is more extensive than the initial one.  The second observation, which is a consequence of the first one, is that it could provide solutions that may not belong to IDG. For this reason, the condition (\ref{rel:hp}) must always be satisfied to recover the IDG solutions.  The third observation is that the Lagrangian density is not manifestly covariant due to convolutions. 

\vspace{0.1cm}
\noindent
We define the Euler-Lagrange equations as
\begin{equation}
\label{eqdif}
\int_{\mathbb{R}^4} \D x\:\lambda_{\mu\nu}(\Phi_{\alpha\beta},x,y) = - \mathcal{J}_{\mu\nu}(y)
\end{equation}
where $\mathcal{J}_{\mu\nu}(y)$ is a source coming from a Lagrangian density of matter $\mathcal{L}_M(\Phi_{\alpha\beta})$.  Having in mind $y^\mu=(\tau,\mathbf{y})$ and using (\ref{lambda}), we get
\begin{equation}
\label{Lambda}
\begin{split}
\lambda_{\mu\nu}(\Phi_{\alpha\beta},x,y) = \frac{1}{4\kappa}&\left\{M_{\mu\nu}^{\;\;\;\;\alpha\beta\sigma\rho}\:\delta(t-\tau)\:\calE_{(\ell)}(\mathbf{x}-\mathbf{y})\:\partial_\alpha\partial_\beta\left(\calT_{(\ell)}\ast\Phi_{\sigma\rho}\right)(x) \right.\\
&\left. + M^{\sigma\rho\alpha\beta}_{\;\;\;\;\;\;\;\;\mu\nu}\left(\calE_{(\ell)}\ast\Phi_{\sigma\rho}\right)(x)\:\partial_\alpha\partial_\beta\left[\calT_{(\ell)}(t-\tau)\:\delta(\mathbf{x}-\mathbf{y})\right]\right\}\:.
\end{split}
\end{equation}
Therefore, the equations of motion are
\begin{equation}
\label{IntEOM2}
\left(M_{\mu\nu}^{\;\;\;\;cd\sigma\rho} + M^{\sigma\rho\alpha\beta}_{\;\;\;\;\;\;\;\mu\nu} \right)\partial_\alpha\partial_\beta\left( \calTE_{(\ell)}\ast\Phi_{\sigma\rho}\right)(y) = -4\kappa\:\mathcal{J}_{\mu\nu}(y)
\end{equation}
where $\calTE_{(\ell)}(y):=  \calT_{(\ell)}(t) \calE_{(\ell)}(\mathbf{y})$ for ease of notation. Recall that coordinating the kinematic space with the functions $\Phi_{\mu\nu}(x)$ implies that the  trace-reversed function $\hat{h}_{\mu\nu}(x)$ is modified as follows 
\begin{equation}
\label{Ptrans}
\hat{\Phi}_{\mu\nu}(x) : = \hat{h}_{\mu\nu} (x) = \left(\calE_{(\ell)}\ast\Phi_{\mu\nu}\right)(x) - \frac{1}{2} \eta_{\mu\nu}\left(\calE_{(\ell)}\ast \Phi\right)(x)
\end{equation}
where $\Phi(x)= \eta^{\mu\nu}\Phi_{\mu\nu}(x)$. Consequently, the gauge condition is
\begin{equation}
\label{gcond}
\partial_{\mu}\hat\Phi^{\mu\nu}(x) :=\partial_\mu \hat h^{\mu\nu} (x)= \left(\partial^\alpha \eta^{\beta\nu} - \frac{1}{2}\partial^\nu \eta^{\alpha\beta}\right)\left(\calE_{(\ell)}\ast\Phi_{\alpha\beta}\right)(x) = 0\:.
\end{equation}
Noticing that $\hat \Phi := \hat h = - \left(\calE_{(\ell)}\ast\Phi\right)(x)$ and using (\ref{Ptrans}) and (\ref{gcond}), the equations of motion simplify as follows
\begin{equation}
\label{IntEOM}
\Box \left(\calT_{(\ell)}\ast\hat\Phi_{\mu\nu}\right)(y) = - 2\kappa\:\calJ_{\mu\nu}(y)\:.
\end{equation} 

\begin{proposition}
\label{rel}
The relation between the Hilbert energy-momentum tensor for matter $T_{\mu\nu}(x)$ and the source $\mathcal{J}_{\mu\nu}(x)$ is 
\begin{equation}
\mathcal{J}_{\mu\nu}(x) = \left(\calE_{(\ell)}\ast T_{\mu\nu}\right)(x)\:.
\end{equation}
\end{proposition}

\begin{proof}
We know that the relation between $h_{\mu\nu}(x)$ and $\Phi_{\mu\nu}(x)$ is given by (\ref{rel:hp}). Thus, having in mind (\ref{T}), $y^\mu = (\tau,\mathbf{y})$ and $z^\mu = (\xi,\mathbf{z})$,
\begin{equation}
\begin{split}
\int_{\mathbb{R}^4} \D x \frac{\delta\mathcal{L}(\Phi_{\alpha\beta},x)}{\delta \Phi^{\mu\nu}(y)} &= \int_{\mathbb{R}^8} \D x \D z \frac{\delta h_{\sigma\rho}(z)}{\delta \Phi^{\mu\nu}(y)}\frac{\delta\mathcal{L}(h_{\alpha\beta},x)}{\delta h^{\sigma\rho}(z)}\\
& = - \int_{\mathbb{R}^3} \D \mathbf{z} \:\mathcal{E}_{(\ell)}(\mathbf{y}-\mathbf{z})\:T_{\mu\nu}(\tau,\mathbf{z})\:,
\end{split}
\end{equation}
and therefore,
\begin{equation}
\mathcal{J}_{\mu\nu}(x) = \left(\calE_{(\ell)}\ast T_{\mu\nu}\right)(x)\:.
\end{equation}
\end{proof}
According to \cite{Vladimirov}, the solution $\hat\Phi_{\mu\nu}(x)$ of the integro-differential equation (\ref{IntEOM}) might be expressed as,
\begin{equation}
\label{soldif}
\hat \Phi_{\mu\nu}(y) = 2\kappa \:(G\ast \mathcal{J}_{\mu\nu}) (y)
\end{equation}
where the kernel $G$ is the fundamental solution of the integrodifferential operator $\Box(\calT_{(\ell)}\ast -)$ belonging to $\mathcal{S}^\prime$ that satisfies
\begin{equation}
\label{GM}
\Box (\calT_{(\ell)}\ast G)(y) = - \delta(y)\:.
\end{equation}
This kernel is commonly known as the Green function in physics \cite{boos2020}. This solution exists as long as $G\ast\mathcal{J}$ and $\calT\ast G\ast\mathcal{J}$ exist\footnote{For the reader who is more familiar with these concepts, the theorem is based on the space in the distributions that acts on compact support functions, namely, on $D^\prime$. Nevertheless, we know that $\mathcal{S}^\prime\subset D^\prime$, so if it is true in $D^\prime$ it will also be true of $\mathcal{S}^\prime$. See more details in \cite{Vladimirov}.} in $\mathcal{S}^\prime$. Recall that the fundamental solution $G$ is not unique, since we can always add a solution $G_0$ such that it satisfies the homogeneous integro-differential equation $\Box\left(\calT\ast G_0\right) = 0$. Thus, the solution is unique as long as $G_0$ is determined.

\section{Euler-Lagrange's solutions}
\label{sectionV}

In order to obtain the complete solution, we need to find the fundamental solution of the integro-differential operator  $\Box (\calT_{(\ell)}\ast -)$ in $\calS^\prime$. For that, we shall previously determine the homogeneous $G_0$ and the inhomogeneous $G_{I}$ of the integro-differential equation (\ref{GM}). The mathematical tool to find them will be the Fourier transform in $\calS^\prime$.

\subsection{Solution of the homogeneous integro-differential  equation}
To find the solution for the homogeneous integro-differential  equation (\ref{GM}), we need to solve 
\begin{equation}
 \Box (\calT_{(\ell)}\ast G_0)(y) = 0\;.
\end{equation}
Applying the Fourier transform in $\mathcal{S}^\prime(\mathbb{R}^{4})$, we get
\begin{equation}
(\omega - |\mathbf{k}|)(\omega+|\mathbf{k}|) \tilde{G}_0(\omega,\mathbf{k}) = 0 \;.
\end{equation}
Invoking the theorem in which is shown that, if the support of a distribution $f$ is the point $\{0\}$, then it is uniquely representable in the form \cite{Vladimirov}
\begin{equation}
f(x) = \sum^{n}_{j=1} \sum^{r_j -1}_{k=0} C_{jk} \delta^{(k)}(x - x_j)
\end{equation}
where $n$ is the number of zeros, $r_j$ is the periodicity of each one and $C_{jk}$ are real constants, we obtain
\begin{equation}
\tilde{G}_0(\omega,\mathbf{k}) = A \delta(\omega -|\mathbf{k}|) + B \delta(\omega + |\mathbf{k}|) 
\end{equation}
where $A:= C_{10}$ and $B:=C_{20}$. Therefore,
\begin{equation}
G_0(t,\mathbf{x}) = \frac{1}{(2\pi)^4} \int_{\mathbb{R}^3} \D\mathbf{k} \left\{A e^{i( |\mathbf{k}| t + \mathbf{k}\cdot\mathbf{x})} + B e^{-i(|\mathbf{k}|t - \mathbf{k}\cdot\mathbf{x})} \right\} \;.
\end{equation}

\subsection{Solution of the inhomogeneous integro-differential  equation}
Let us discuss the inhomogeneous integro-differential  equation (\ref{GM}). Applying the Fourier transformation in $\mathcal{S}^\prime(\mathbb{R}^4)$,  we get
\begin{equation}
\label{eqf}
\begin{split}
\tilde{G}_{I}(\omega,\mathbf{k}) = - \frac{e^{\ell^2\omega^2}}{\omega^2-|\mathbf{k}|^2}\;. 
\end{split}
\end{equation}
Note that isolating $\tilde{G}_{I}(\omega,\mathbf{k})$ causes that the term on the right-side of equation (\ref{eqf}) does not belong to $\mathcal{S}^\prime(\mathbb{R}^4)$ since 
\begin{equation}
\begin{split}
\left(\mathcal{F}^{-1}\left[\tilde{G}_{I}(\omega,\mathbf{k})\right], \varphi\right) &= \left(\tilde{G}_{I}(\omega,\mathbf{k}),\mathcal{F}^{-1}\left[\varphi\right]\right)\\
& = \frac{1}{(2\pi)^4}\int_{\mathbb{R}^4}\D x\:\varphi(x) \int_{\mathbb{R}^4}\D\omega \D\mathbf{k}\:\tilde{G}_{I}(\omega,\mathbf{k}) e^{i(\omega t + \mathbf{k}\cdot\mathbf{x})} 
\end{split}
\end{equation}
diverges. Therefore, the inverse of Fourier transform does not exist, and consequently, we cannot find $G(x)$.  This suggests that the kernel might not be a locally integrable function of slow growth. In fact, obtaining the kernel $G_{I}(x)$ for any $\mathcal{J}_{\mu\nu}(x)$, the initial value problem would be wholly determined since the theorem assures that the solution (\ref{soldif}) is unique because $G_0(x)$ is completely determined.

\vspace{0.1cm}
Due to the complexity of this initial value problem, we will adopt the point of view raised in \cite{Llosa1994,Heredia2,Gomis2001} and we will assume that the Euler-Lagrange equations (\ref{IntEOM}) are not longer dynamic equations but constraints that limit the kinematic space $\mathcal{K}$. We will search what conditions $T_{\mu\nu}(x)$ shall meet to obtain a solution.  We will work on the static case and the homogeneous case separately.

\subsubsection{Static case}
Let us focus only on the static case $\hat\Phi_{\mu\nu}(t,\mathbf{x}) := \hat\Phi_{\mu\nu}(\mathbf{x})$. Under this assumption, equation (\ref{IntEOM}) becomes
\begin{equation}
\Delta \hat \Phi_{\mu\nu}(\mathbf{y})=  - 2\kappa \left(\calE_{(\ell)}\ast T_{\mu\nu}\right)(\mathbf{y})
\end{equation}
where we have used proposition (\ref{rel}). Using proposition (\ref{FourierConv}), the solution of this integro-differential equation is 
\begin{equation}
\label{solsp}
\hat\Phi_{\mu\nu}(\mathbf{y}) = 2 \kappa\:\mathcal{F}^{-1}\left[\frac{e^{-\ell^2|\mathbf{k}|^2}}{|\mathbf{k}|^2} \tilde T_{\mu\nu}\right](\mathbf{y}),
\end{equation}
which only exists in $\mathcal{S}^\prime(\mathbb{R}^3)$ if, and only if, the inverse of Fourier transform exists in $\mathcal{S}^\prime(\mathbb{R}^3)$. Therefore, 
not all sources $T_{\mu\nu}(\mathbf{x})$ can be used to find a solution in $\mathcal{S}^\prime(\mathbb{R}^3)$, but only those that satisfy that $\frac{e^{-\ell^2|\mathbf{k}|^2}}{|\mathbf{k}|^2} \tilde T_{\mu\nu}$ belongs to $\calS^\prime(\RR^3)$. Note that we cannot use proposition (\ref{FourierConv}) again to compute the inverse of the Fourier transform, since the factor $\frac{e^{-\ell^2|\mathbf{k}|^2}}{|\mathbf{k}|^2}\in\calS^\prime(\RR^3)$ and $\tilde T_{\mu\nu}\in\calS^\prime(\RR^3)$. Let us explore a sufficient condition so that the inverse Fourier transform exists. 

\begin{remark}
\label{Tsol_2}
A bounded function $\mathcal{H}(\mathbf{x})$ belongs to $\calS^\prime(\RR^3)$ as a regular distribution. 
\end{remark}

\begin{proof}
To prove this remark, we will rely on proposition (\ref{defint}); Showing that $\mathcal{H}(\mathbf{x})$ is a locally integrable function of slow growth, we can assure it belongs to $\calS^\prime(\RR^3)$ as a regular distribution.

\vspace{0.1cm}
\noindent
According to our assumption $\mathcal{H}(\mathbf{x})$ is bounded. Therefore, there exists a number $M > 0$, such that
\begin{equation}
\underset{\mathbf{k}\in\mathbb{R}^3}{\mathrm{sup}}\left|\mathcal{H}(\mathbf{x})\right| = M.
\end{equation}
Thus,
\begin{equation}
\begin{split}
\int_{\mathbb{R}^3} \D\mathbf{x}   \left|\mathcal{H}(\mathbf{x})\right| (1+|\mathbf{x}|^2)^{-\frac{s}{2}} & \leq C  \underset{\mathbf{x}\in\mathbb{R}^3}{\mathrm{sup}} \left|\mathcal{H}(\mathbf{x})\right| \int_{\mathbb{R}^3} \frac{\D\mathbf{x}}{(1+|\mathbf{x}|^2)^{\frac{s}{2}}} \leq C_1 \underset{\mathbf{x}\in\mathbb{R}^3}{\mathrm{sup}} \left|\mathcal{H}(\mathbf{x})\right| < \infty
\end{split}
\end{equation}
where $C_i$ are positive real constants. 
\end{proof}

\noindent
As a consequence of this remark, 
\begin{proposition}
\label{Tsol}
Let $\mathcal{H}(\mathbf{k})$ be a bounded function. It is sufficient that $\tilde T_{\mu\nu}(\mathbf{k})$ is
\begin{equation}
\label{Jsolk}
\tilde T_{\mu\nu}(\mathbf{k}) = C_{\mu\nu} \mathcal{H}(\mathbf{k}) 
\end{equation}
where $C_{\mu\nu}$ is a constant tensor so that solution (\ref{solsp}) exits.
\end{proposition}

\begin{proof}
As before,  we will rely on proposition (\ref{defint}); By showing that $\frac{e^{-\ell^2|\mathbf{k}|^2}}{|\mathbf{k}|^2} \tilde T_{\mu\nu}$ is a locally integrable function of slow growth, we can confirm that it belongs to $\calS^\prime(\RR^3)$, and consequently, the inverse of the Fourier transform will exist.  Therefore, according to our assumption $\mathcal{H}(\mathbf{k})$ is bounded; there exists a number $M > 0$, such that
\begin{equation}
\underset{\mathbf{k}\in\mathbb{R}^3}{\mathrm{sup}}\left|\mathcal{H}(\mathbf{k})\right| = M.
\end{equation}
Thus,
\begin{equation}
\begin{split}
\int_{\mathbb{R}^3} \D\mathbf{k}  \frac{e^{-\ell^2|\mathbf{k}|^2}}{|\mathbf{k}|^2} \left|\tilde T_{\mu\nu}(\mathbf{k})\right| (1+|\mathbf{k}|^2)^{-\frac{s}{2}} &=C_{\mu\nu} \int_{\mathbb{R}^3} \D\mathbf{k} \frac{\left|\mathcal{H}(\mathbf{k})\right|}{|\mathbf{k}|^2(1+|\mathbf{k}|^2)^{\frac{s}{2}}} e^{-\ell^2|\mathbf{k}|^2}\\
& \leq C_{\mu\nu}  \underset{\mathbf{k}\in\mathbb{R}^3}{\mathrm{sup}} \left|\mathcal{H}(\mathbf{k})\right| \int_{\mathbb{R}^3} \D\mathbf{k}\frac{ e^{-\ell^2|\mathbf{k}|^2}}{|\mathbf{k}|^2(1+|\mathbf{k}|^2)^{\frac{s}{2}}} \frac{(1+|\mathbf{k}|^2)^{\frac{s}{2}}}{(1+|\mathbf{k}|^2)^{\frac{s}{2}}}\\
& \leq C_{(1)\mu\nu} \underset{\mathbf{k}\in\mathbb{R}^3}{\mathrm{sup}} e^{-\ell^2|\mathbf{k}|^2}(1+|\mathbf{k}|^2)^{\frac{s}{2}}  \int_{\mathbb{R}^3} \frac{\D\mathbf{k}}{|\mathbf{k}|^2(1+|\mathbf{k}|^2)^s}\\
&  \leq C_{(2)\mu\nu}\int^\infty_0 \frac{\D r}{(1+r^2)^s} <\infty
\end{split}
\end{equation}
where $C_{(i)\mu\nu}$ are constant tensors and the spherical coordinates have been used.
\end{proof}

\noindent
Let us illustrate this proposition with three examples.
 
\paragraph*{First example}
The simplest example of a bounded Hilbert energy-momentum tensor in Fourier space is 
\begin{equation}
\label{TE1}
\tilde T_{\mu\nu}(\mathbf{k}) = m \delta^0_\mu \delta^0_\nu\:\mathcal{H}(\mathbf{k}) \qquad \mathrm{with} \qquad \mathcal{H}(\mathbf{k}) = 1 \:.
\end{equation}
Applying the inverse of the Fourier transform in $\calS^\prime(\RR^3)$, we obtain that the Hilbert energy-momentum tensor for matter $T_{\mu\nu}(\mathbf{x})$ is
\begin{equation}
T_{\mu\nu} (\mathbf{x}) = m  \delta^0_\mu \delta^0_\nu\:\delta (\mathbf{x}),
\end{equation}
which is the habitual Delta Dirac source. Using (\ref{solsp}) with (\ref{TE1}), we get 
 \begin{equation}
\label{h3}
\hat \Phi_{00}(\mathbf{x}) = \frac{\kappa m}{2\pi}\frac{1}{ |\mathbf{x}|}\mathrm{Erf}\left[\frac{|\mathbf{x}|}{2\ell}\right]\:.
\end{equation}
This result was found in \cite{Biswas2012}.

\paragraph*{Second example}
We take as another possibility
\begin{equation}
\label{TE3}
\tilde T_{\mu\nu}(\mathbf{k}) = m\:\delta^0_\mu \delta^0_\nu \:\mathcal{H}(\mathbf{k}) \qquad \mathrm{with} \qquad \mathcal{H}(\mathbf{k}) = J_0\left(\sqrt{k^2_x + k^2_y}\:a\right)
\end{equation}
where $J_0(x)$ is the Bessel function --- bounded function --- , $a$ is a real positive constant, and we assume that $\mathbf{k}^2 = k_x^2+k_y^2+k_z^2$. Applying the inverse of Fourier transform using the cylindrical coordinate\footnote{The cylindrical coordinates are: $x = r \cos(\theta)$, $y = r \sin(\theta)$ and $z =z$, and the volume element is $\D\mathbf{x} = r \D r \D\theta \D z$. } in $\mathcal{S}^\prime(\mathbb{R}^3)$, we find that 
\begin{equation}
\label{TE2}
T_{\mu\nu}(\mathbf{x}) = \frac{m}{\pi} \delta^0_\mu \delta^0_\nu\:\delta(z) \delta(x^2 +y^2 -a^2) 
\end{equation} 
is the Hilbert energy-momentum tensor for a distribution of matter on a ring of radius $a$. Notice that we have used the following property of Bessel functions
\begin{equation}
\int^\infty_{0} \D r\:r\:J_{\alpha}\left(\beta r\right) J_\alpha\left(\xi r\right) =  \frac{1}{\xi} \delta(\xi-\beta)\:.
\end{equation}
Plugging equation (\ref{TE3}) into (\ref{solsp}) and setting $z=0$,  the solution becomes
\begin{equation}
\hat{\Phi}_{00}(\mathbf{x}) = \frac{\kappa m}{2\pi} \int^\infty_{0} \D r J_0\left(a r\right) J_0\left(\sqrt{x^2 + y^2} \:r\right) \mathrm{Erfc}(\ell r)
\end{equation}
where $\mathrm{Erfc}(x)$ is the complementary error function. This result was found in \cite{Buoninfante2018}. 

\paragraph*{Third example}
Note that proposition (\ref{Tsol}) indicates a sufficiency condition. It may therefore be the case that the Hilbert energy-momentum tensor is described in another way, for instance, by the Dirac delta or the Principal Value $\mathcal{P}$. A very illustrative example is the following: 
\begin{equation}
\label{Jcosmic}
\tilde T_{\mu\nu}(\mathbf{k}) = \mathcal{Y}_{\mu\nu}(k_s) \left( i\:\mathcal{P}\left[\frac{1}{k_z}\right] +\pi \delta(k_z)\right)
\end{equation}
where $\mathcal{Y}_{\mu\nu}$ is a tensor whose components are bounded functions of the s-component of the vector $k_s = (k_x,k_y)$. We choose for convenience that $\mathcal{Y}_{\mu\nu}(k_s)$ is 
\begin{equation}
\label{Pmn}
\mathcal{Y}_{\mu\nu}(k_s) := - i D^n_{\mu\nu} j^{\:s}_{n} k_s
\end{equation}
where $D^n_{\mu\nu}$ is a constant tensor, $j_{sn}$ is an antisymmetric angular tensor, and the Latin indices $s,n$ are two-dimensional indices of the $xy$-plane. Applying the inverse Fourier transform in $\mathcal{S}^\prime(\mathbb{R}^3)$ to (\ref{Jcosmic}) with (\ref{Pmn}) and
\begin{equation}
\label{CD}
D^n_{\mu\nu} = \delta^0_{(\mu}\delta^n_{\nu)}\:,
\end{equation}
we obtain that 
\begin{equation}
\label{Tcosmic}
T_{\mu\nu}(\mathbf{x}) =  -\delta^0_{(\mu}\delta^n_{\nu)} j^{\:s}_n \partial_s \delta(x) \delta(y) \theta(z) 
\end{equation}
is the Hilbert energy-momentum tensor for a spinning semi-infinite string. One might be tempted to plug equation (\ref{Jcosmic}) into (\ref{solsp}) and calculate the inverse of Fourier transform. However, this is not the most optimal way to do it. To compute $\hat \Phi_{\mu\nu}(\mathbf{x})$, we will proceed in the same way it is stated in the appendix of \cite{Boos2021}. This method consists of deriving with respect to $z$ the Hilbert energy-momentum tensor $T_{\mu\nu}(\mathbf{x})$ in order to transform the Heaviside theta into a Dirac delta and thus get the solution from the standard point source. Due to its simplicity, we will adapt this idea and apply it within this approach. Therefore, we derive with respect to $z$ equation (\ref{Tcosmic}),
\begin{equation}
\label{php}
T^\prime_{\mu\nu}(\mathbf{x}) =  -\delta^0_{(\mu}\delta^n_{\nu)} j^{\:s}_n \partial_s \delta(\mathbf{x})
\end{equation}
where $'$ denotes the derivative with respect to $z$. Consequently, solution (\ref{solsp}) becomes 
\begin{equation}
\label{hfd}
\begin{split}
\hat \Phi^\prime_{\mu\nu}(\mathbf{y}) &= - 2 i \kappa\:\delta^0_{(\mu}\delta^n_{\nu)} j^{\:s}_n\:\calF^{-1}\left[k_s \frac{e^{-\ell^2 |\mathbf{k}|^2}}{|\mathbf{k}|^2}\right] \\
&= - \frac{\kappa}{2\pi} \delta^0_{(\mu}\delta^n_{\nu)} j^{\:s}_n\:\partial_s\left(\frac{1}{|\mathbf{x}|} \Erf\left[\frac{|\mathbf{x}|}{2\ell}\right]\right)\:.
\end{split}
\end{equation}
Finally, by integrating over $z$, we get
\begin{equation}
\begin{split}
\hat \Phi_{\mu\nu}(\mathbf{x}) &= \int^z_{-\infty} \D \tilde z\:\hat\Phi^\prime_{\mu\nu}(\mathbf{x}) \\
& = \frac{\kappa}{2\pi} \delta^0_{(\mu}\delta^n_{\nu)} j^{\:s}_n\left\{1- \left(1+ \Erf\left[\frac{z}{2\ell}\right]\right)e^{-\frac{|\mathbf{x}|^2_{_{(2)}}}{4\ell^2}}+ \frac{z}{|\mathbf{x}|}\Erf\left[\frac{|\mathbf{x}|}{2\ell}\right]\right\}
\end{split}
\end{equation}
where $|\mathbf{x}|_{(2)} := \sqrt{x^2+y^2}$ and we have first derived and then integrated.  This result was found in \cite{Kol2020}.

\vspace{0.3cm}
\noindent
With this example, it is clear that there may be cases where the Hilbert energy-momentum tensor can be described differently, and still, the solution exists. However, note that $\frac{e^{-\ell^2|\mathbf{k}|^2}}{|\mathbf{k}|^2} \tilde T_{\mu\nu}\in\calS^\prime(\RR^3)$ is still fulfilled for this case.

\vspace{0.3cm}
To finish these three examples, one would have to obtain $h_{\mu\nu}(\mathbf{x})$ from $\hat\Phi_{\mu\nu}(\mathbf{x})$. However, thanks to definition (\ref{Ptrans}) and condition (\ref{rel:hp}), for this particular case it is straightforward, since we can directly relate $\hat\Phi_{\mu\nu}(\mathbf{x})$ with $h_{\mu\nu}(\mathbf{x})$, i.e, 
\begin{equation}
h_{\mu\nu}(\mathbf{x}) = \hat\Phi_{\mu\nu}(\mathbf{x}) - \frac{1}{2} \eta_{\mu\nu}\hat\Phi(\mathbf{x})\:.
\end{equation}

\subsubsection{Homogeneous time-dependent case}

Let us mention a few words about the homogeneous time-dependent case $\hat\Phi_{\mu\nu}(t,\mathbf{x}) := \hat\Phi_{\mu\nu}(t)$. For this particular case, proposition (\ref{rel}) do not affect and equation (\ref{IntEOM}) then simplifies as follows
\begin{equation}
\partial^2_t(\calT_{(\ell)}\ast \hat\Phi_{\mu\nu}) (t) = 2\kappa \:T_{\mu\nu}(t)
\end{equation}
where the Hilbert energy-momentum tensor for matter $T_{\mu\nu}(t)$ is supported on $\mathbb{R}^+$. The solution of this integral equation in $\mathcal{S}^\prime(\mathbb{R}^+)$ is
\begin{equation}
\hat \Phi_{\mu\nu}(t) = - 2 \kappa\:\mathcal{F}^{-1}\left[\frac{e^{\ell^2\omega^2}}{\omega^2}\tilde T_{\mu\nu}\right](t)
\end{equation}
which only exists if, and only if, the inverse of Fourier transform exists in $\mathcal{S}^\prime(\mathbb{R}^+)$. Therefore, we need that
\begin{equation}
\label{Homocase}
\frac{e^{\ell^2\omega^2}}{\omega^2}\tilde T_{\mu\nu}\in\calS^\prime(\RR^+)\:.
\end{equation}
As we have discussed above, not all Hilbert energy-momentum tensors for matter can be used to obtain a solution in $\mathcal{S}^\prime(\mathbb{R})$, but only those that satisfy the condition (\ref{Homocase}). However, let us note that this condition might be fulfilled as long as the Hilbert energy-momentum tensor is non-local scale-dependent 
\begin{equation}
\tilde T_{\mu\nu}(\omega) \sim e^{-a^2 \omega^2}, \qquad \mathrm{with} \qquad a\in\RR^+ \qquad \mathrm{and} \qquad a^2 > \ell^2   
\end{equation}
since, it could cancel the exponential growth. From a mathematical point of view, this source provides a correct solution; however, from a physical point of view, the dependence on the non-local scale parameter in the Hilbert energy-momentum tensor makes us discern that the solution obtained does not belong to IDG. Recall that non-locality is introduced only via form-factors (which depend on $\ell$), not modified sources. Thus, the interpretation of the sources with dependence on $\ell$ is not clear in this physical context. Therefore, following this argument (i.e., sources that may depend on the scale are not considered), the condition is not satisfied since $\frac{e^{\ell^2\omega^2}}{\omega^2}\tilde T_{\mu\nu}$ is not a tempered distribution. Consequently, the Fourier transform does not exist in $\mathcal{S}^\prime(\mathbb{R}^+)$ and therefore, the solution does not exist in $\mathcal{S}^\prime(\mathbb{R}^+)$ either .

\section{Conclusion}
\label{sectionVI}

In IDG, non-locality is introduced by infinite derivatives encapsulated in forms factors. These form factors are entire functions that do not include any additional zeros in the complex plane, and might be then expressed by formal Taylor series. Handling these infinite series requires controlling their convergence when they act on functions belonging to class $\calC^\infty(\RR^4)$.  Failure to do so could lead to the mistake of working with series that could be infinite. 

\vspace{0.1cm}
Throughout this article, we focused only on the ghost-free non-local operator and a very particular theory of linearized gravity in which the Lagrangian density is entirely analytical. In order to avoid working with infinite series, we transformed the infinite-order Lagrangian into a non-local one. We achieved this transformation due to the definition of the non-local operator acting on functions $\mathcal{C}^\infty(\mathbb{R}^4)$ via the inverse Fourier transform in the space of tempered distributions $\mathcal{S}^\prime(\RR^4)$. 

\vspace{0.1cm}
This way of defining how the non-local operator acts on these functions has two significant consequences. The first consequence is the indentification of the kinematic space $\calK$. As we showed, condition (\ref{rel:hp}) must be satisfied, so that proposition (\ref{def:noperator}) holds. Consequently, the non-local operator is only well-defined on a subset of it. The second consequence (derived from the first one) is the structure of these functions. As we proved --- propositions (\ref{bound}) and (\ref{bound2}) ---,  the tensor field $h_{\mu\nu}(x)$ and its derivatives for which the non-local operator is well-defined are bounded by a polynomial in the static case.  Therefore, they are regular in all $\RR^3$. This fact is of great importance since it indicates that the Riemann tensor --- where the second derivatives of $h_{\mu\nu}$ are involved --- will be finite everywhere and, consequently, all the scalar curvature invariants will be finite as well. The latter is an indicator that, under these conditions, there is no problem of spacetime singularities. 

\vspace{0.1cm}
Transforming the infinite-order Lagrangian into a non-local one causes that the new variables $\Phi_{\mu \nu}(x)$ are no longer the gravitational fields $h_{\mu\nu}(x)$, and their sources are not the Hilbert energy-momentum tensor for matter $T_{\mu\nu}(x)$ but a scaled tensor $\mathcal{J}_{\mu\nu}(x)$. Consequently, we could find solutions that do not belong to IDG. For this reason, we always need to keep in mind the condition (\ref{rel:hp}) to validate whether the solution belongs to IDG or not. Another consequence of this method is that the equations of motion to be solved are inhomogeneous linear integro-differential equations. According to \cite{Vladimirov}, the general solution of this problem might be then expressed by employing the convolution between the source and the fundamental solution of the operator of the inhomogeneous linear integro-differential equation. However, we observed that the fundamental solution of this integro-differential operator does not exist in $\mathcal{S}^\prime(\mathbb{R}^4)$, implying that it might not be a locally integrable function of slow growth.  

\vspace{0.1cm}
Due to the complexity of the initial value problem, we followed the idea of \cite{Llosa1994,Heredia2,Gomis2001} where Euler-Lagrange's equations are taken as constraints that restrict the kinematic space $\mathcal{K}$ rather than dynamic equations. For simplicity, we worked on the static case and the homogeneous time-dependent case separately. Under this approach, we observed that, for the homogeneous time-dependent case, if the Hilbert energy-momentum tensor were non-local scale-dependent, the solution would exist, since condition (\ref{Homocase}) could be fulfilled. However, as the non-locality in IDG is introduced by form factors and not by modified sources, we discern that these solutions might belong to it. Therefore, considering the case where the Hilbert energy-momentum tensor does not depend on the non-local scale, the solution may not exist in $\mathcal{S}^\prime$. On the other hand, we showed that, for the static case, we could build solutions in $\mathcal{S}^\prime$ if the Hilbert energy-momentum tensor in the Fourier space is a bounded function. To exemplify it, we discussed three examples.

\vspace{0.1cm}
Finally, let us highlight that this mathematical treatment of non-local operators applies to linearized IDG and full IDG, if one focuses on geometries for which the entire field equations reduce to linear equations. This fact happens, for instance, for the class of almost universal spacetimes \cite{Kolar2021_4,Kolar2021_3}. On the other hand, note that not having been able to find a general solution resides in how we defined the non-local operator. This definition limits in excess the type of solution we might have. For this reason, one should study further (for instance, using other definitions) how this operator might act on functions of class $\mathcal{C}^\infty(\mathbb{R}^4)$ to keep investigating new solutions for these kinds of theories (see \cite{Gruman1973,Carlsson2014} and references therein). Moreover, it would be very challenging to consider how to generalize proposition (\ref{def:noperator}) in a covariant way.  We will leave it for further investigation.

\section*{Acknowledgment}
Authors would like to thank Joaquim Gomis for valuable discussions about non-local theories. J.LL. was supported by the Spanish MINCIU and ERDF (project ref. RTI2018-098117-B-C22). I.K. and A.M. were supported by Netherlands Organization for Scientific Research (NWO) grant no. 680-91-119. F.J.M.T. acknowledges financial support from NRF grants no. 120390, reference: BSFP190416431035; no. 120396, reference: CSRP190405427545; no. 101775, reference: SFH150727131568, and the support from the Research Council of Norway.

\appendix
\section*{Appendix}
\section{Space of tempered distributions $\mathcal{S}^\prime$}
\label{A:S}
For completeness of this article, let us specify some definitions and review some concepts of tempered distributions \cite{Vladimirov_GF,Vladimirov,Hormander1990}.

\begin{definition}
\label{def:S}
$\mathcal{S}$ is the class of all functions $f\in\mathcal{C}^\infty(\mathbb{R}^n)$ such that they decay faster than any power of $|x|^{-1}$. The norm defined in $\mathcal{S}$ is
\begin{equation}
||f||_s = \underset{|\alpha| \leq s, x\in\mathbb{R}^n}{\mathrm{sup}} (1+ |x|^2)^{\frac{s}{2}} \left|D^\alpha f(x)\right| 
\end{equation}
with $s\in\mathbb{N}$ and 
\begin{equation}
\label{Deriv}
D^\alpha f(x)= \frac{\partial^{|\alpha|} f(x)}{\partial x^{\alpha_1}_1 ... \partial x^{\alpha_n}_n}
\end{equation}
where $\alpha$ denotes the multi-index notation. This space is known as \textit{the space of rapidly diminishing functions} or \textit{Schwartz space}\cite{Schwartz1966}.
\end{definition}

\begin{proposition}
\label{SsL}
Functions belonging to $\mathcal{S}$ are always $L^p$-integrable functions, i.e. $\mathcal{S}(\mathbb{R}^n) \subset L^p(\mathbb{R}^n)$.
\end{proposition}
\begin{proof}
Let $f\in\mathcal{S}(\mathbb{R}^n)$. For a $s\geq 0$, $|f(x)| (1+|x|^2)^{\frac{s}{2}}$ is bounded, namely
\begin{equation}
(1+|x|^2)^{\frac{s}{2}} |f(x)|  \leq \underset{s\geq0, x\in\mathbb{R}^n}{\mathrm{sup}}(1+|x|^2)^{\frac{s}{2}}|f(x)| \;.
\end{equation}
Therefore,
\begin{align}
\int_{\mathbb{R}^n}\D x\:|f(x)|^p &= \int_{\mathbb{R}^n}\D x\:\left[(1+|x|^2)^{\frac{s}{2}}|f(x)|\right]^p (1+|x|^2)^{-\frac{sp}{2}}\nonumber\\
&\leq \int_{\mathbb{R}^d}\D x\:\left[(1+|x|^2)^{\frac{s}{2}}|f(x)|\right]^p (1+|x|^2)^{-\frac{s}{2}}\nonumber\\
&\leq \left[\underset{s\geq0, x\in\mathbb{R}^d}{\mathrm{sup}}(1+|x|^2)^{\frac{s}{2}}|f(x)|\right]^p \int_{\mathbb{R}^n} \D x \frac{1}{(1+|x|^2)^{\frac{s}{2}}}\nonumber\\
& = C ||f||_s^p < \infty
\end{align}
as long as $s > n$, and $C$ is a real positive constant.
\end{proof}

\begin{remark}
\label{A2}
The gaussian function $f(x) =e^{-a |x|^2}$ belongs to $\mathcal{S}(\mathbb{R}^n)$.
\end{remark}

\begin{proof}
We know that 
\begin{equation}
D^\alpha_x e^{-a |x|^2} = P_{(\alpha)}(x)e^{-a |x|^2} 
\end{equation}
where $P_{(\alpha)}(x)$ is a polynomial of degree $\alpha$. Moreover, we know that a polynomial can always be bounded by another polynomial of a higher degree, namely
\begin{equation} 
\left|P_{(\alpha)}(x)\right| \leq C_1 (1+ |x|^2)^{\frac{s}{2}} \quad \mathrm{with} \quad C_1\in\mathbb{R}^{+} 
\end{equation}
where $\frac{s}{2} \geq \alpha$. On the other hand, we observe that
\begin{equation}
e^{|x|^2} = \sum^\infty_{j=0} \frac{(|x|^{2})^j}{j!}\geq C_2 (1 + |x|^2)^s  \quad \mathrm{with} \quad C_2\in\mathbb{R}^{+}
\end{equation}
since all the terms are positive. Consequently, 
\begin{equation}
\label{expcond}
\left|e^{-|x|^2}\right| \leq C_3 \left| \frac{1}{(1+|x|^2)^s}\right| =  C_3 \frac{1}{(1+|x|^2)^s} \quad \mathrm{with} \quad C_3\in\mathbb{R}^{+}.
\end{equation}
Applying definition (\ref{def:S}), we get
\begin{equation}
\begin{split}
(1+|x|^2)^{\frac{s}{2}}|f(x)| &\leq \underset{x\in\mathbb{R}^n; \alpha\leq s}{\mathrm{sup}} (1+ |x|^2)^{\frac{s}{2}} \left| D^\alpha e^{-a|x|^2} \right|\\
&\leq \underset{x\in\mathbb{R}^n; \alpha\leq s}{\mathrm{sup}} (1+ |x|^2)^{\frac{s}{2}} \left|P_{(\alpha)}(x)\right| \left|e^{-a|x|^2}\right|\\
& \leq C_4 \underset{x\in\mathbb{R}^n; \alpha\leq s}{\mathrm{sup}} (1+ |x|^2)^{s} \left|e^{-a|x|^2}\right|\\
& \leq C_5 \underset{x\in\mathbb{R}^n; \alpha\leq s}{\mathrm{sup}} \frac{(1+ |x|^2)^s}{(1+ |x|^2)^s} = C_5 <\infty
\end{split}
\end{equation}
where $C_4$ and $C_5$ are real positive constants.
\end{proof}

\begin{definition}
$\mathcal{S}^\prime$ is the class of all distributions acting on $\mathcal{S}$. This space is known as \textit{the space of tempered distributions}.
\end{definition}

\begin{proposition}
\label{defint}
Every locally integrable function $f(x)$ of slow growth at infinity i.e.
\begin{equation}
\int_{\mathbb{R}^n} \left|f(x)|\right (1+|x|^2)^{-\frac{s}{2}}\D x <\infty
\end{equation}
for a certain $s\geq 0$ defines a regular distribution\footnote{A regular distribution is a distribution that is generated by a function $f(x)$ locally integrable in $\mathbb{R}^n$. On the other hand, a singular distribution cannot be identified with any locally integrable function. The most common examples of singular distribution are the Dirac Delta $\delta(x)$ and the principal value of $1/x$.} $f$ belonging to $\mathcal{S}^\prime$ according to
\begin{equation}
\label{S'}
(f,\varphi) := \int_{\mathbb{R}^n} f(x) \varphi(x) \D x, \qquad \varphi\in \mathcal{S}(\mathbb{R}^n) \;.
\end{equation}
\end{proposition}

\begin{proof}
Let $\varphi\in\mathcal{S}(\mathbb{R}^n)$. Therefore,
\begin{equation}
\label{Sproof}
\begin{split}
\left|\int_{\mathbb{R}^n} f(x) \varphi(x) \D x \right| &\leq \int_{\mathbb{R}^n} \left|f(x)\right|\left|\varphi(x)\right| dx = \int_{\mathbb{R}^n}\left|f(x)|\right (1+|x|^2)^{-\frac{s}{2}} \left|\varphi(x)\right| (1+|x|^2)^{\frac{s}{2}} \D x\\
& \leq \underset{x\in\mathbb{R}^n; s \geq 0}{\mathrm{sup}} (1+ |x|^2)^{\frac{s}{2}} \left|\varphi(x)\right|  \int_{\mathbb{R}^n}\left|f(x)\right| (1+|x|^2)^{-\frac{s}{2}} \D x < \infty \:.
\end{split}
\end{equation}
Because the integral is a linear operator, it implies that (\ref{S'}) is also linear. Let us prove now continuity. Let $\varphi_k$ be a sequence belonging to $\mathcal{S}$ that converges to $0$ as $k\rightarrow\infty$; therefore, having in mind equation (\ref{Sproof}), $(f,\varphi_k) \rightarrow 0, \:\:\forall f$.
\end{proof}

\begin{proposition}
\label{SsSp}
The class of functions belonging to $\mathcal{S}(\mathbb{R}^n)$ can always be considered as tempered distributions, i.e. $\mathcal{S}(\mathbb{R}^n)\subset \mathcal{S}^\prime(\mathbb{R}^n)$.
\end{proposition}

\begin{proof}
To show this proposition, let us prove first of all that $L^p(\mathbb{R}^n) \subset \mathcal{S}^\prime(\mathbb{R}^n)$. Let $f(x)\in L^p(\mathbb{R}^n)$ and $\varphi(x)\in\mathcal{S}(\mathbb{R}^n)$. Therefore,
\begin{equation}
\begin{split}
\left|(f,\varphi)\right| &\leq \int_{\mathbb{R}^n}\D x \left|\varphi(x)\right| (1+ |x|^2)^{\frac{s}{2}} (1+ |x|^2)^{-\frac{s}{2}}\left| f(x)\right|\\
&\leq  \underset{x\in\mathbb{R}^n; s \geq 0}{\mathrm{sup}} (1+ |x|^2)^{\frac{s}{2}} \left|\varphi(x)\right|  \int_{\mathbb{R}^n}\D x \frac{\left| f(x)\right|}{(1+ |x|^2)^{\frac{s}{2}}}\:.
\end{split}
\end{equation}
As we know that $f(x)\in L^p(\mathbb{R}^n)$ and $1/(1+|x|^2)^{\frac{s}{2}}\in L^p(\mathbb{R}^n)$, we can use Hölder's inequality to get
\begin{equation}
\label{LsSp}
\left|(f,\varphi)\right| \leq \underset{x\in\mathbb{R}^n; s \geq 0}{\mathrm{sup}} (1+ |x|^2)^{\frac{s}{2}}  \left|\varphi(x)\right|   \left(\int_{\mathbb{R}^n} \D x |f(x)|^p \right)^\frac{1}{p} \left(\int_{\mathbb{R}^n} \D x \left|\frac{1}{(1+|x|^2)^{\frac{s}{2}}}\right|^q \right)^\frac{1}{q} < \infty \:.
\end{equation}
Linearity is clear since the integral is a linear operator. Let us prove now continuity. Let $\varphi_k$ be a sequence belonging to $\mathcal{S}$ that converges to $0$ as $k\rightarrow\infty$; therefore, having in mind equation (\ref{LsSp}), $(f,\varphi_k) \rightarrow 0, \:\:\forall f$. Now, invoking proposition (\ref{SsL}), we can confirm that $\mathcal{S}(\mathbb{R}^n)\subset\mathcal{S}^\prime(\mathbb{R}^n)$.
\end{proof}

\begin{definition} 
\label{FT}
We denote the Fourier transform and the inverse as
\begin{equation}
\tilde g(k) = \mathcal{F}[g](k) =  \int_{\mathbb{R}^n} \D x\:g(x) e^{-i k\cdot x} \quad \mathrm{and} \quad g(x) =\mathcal{F}^{-1}[\tilde g](x) = \frac{1}{(2\pi)^n} \int_{\mathbb{R}^n} \D k \:\tilde{g}(k) e^{ik\cdot x} \;.
\end{equation}
\end{definition}

\begin{proposition}
\label{propint}
The Fourier transform and the inverse always exist in $\calS(\RR^n)$.
\end{proposition}
\begin{proof}
Let $g(x) \in \mathcal{S}(\mathbb{R}^n)$. Therefore, because of  $\mathcal{S}(\mathbb{R}^n) \subset L^p(\mathbb{R}^n)$,
\begin{equation}
\left| \mathcal{F}[g](k) \right| = \left|\int_{\mathbb{R}^n} \D x\:g(x) e^{-i k\cdot x}\right| \leq \int_{\mathbb{R}^n} \D x\left|g(x) e^{-i k\cdot x}\right| \leq  \int_{\mathbb{R}^n} \D x\left|g(x)\right| < \infty \;.
\end{equation}
For the inverse, the proof is equivalent.
\end{proof}

\begin{definition}
\label{FTS}
Taking into account that the Fourier transform is a linear and continuous map from $\mathcal{S}$ to $\mathcal{S}$ \cite{Vladimirov}, we define the Fourier transform and the inverse for any distribution of the slow growth $g\in\mathcal{S}^\prime(\mathbb{R}^n)$ as
\begin{equation}
(\mathcal{F}[g],\varphi) := (g,\mathcal{F}[\varphi]) \qquad \mathrm{and} \qquad  (\mathcal{F}^{-1}[g],\varphi) := (g,\mathcal{F}^{-1}[\varphi])  \qquad \varphi\in\mathcal{S}(\mathbb{R}^n) \;.
\end{equation}
\end{definition}

\begin{definition}
\label{PropConvSp}
Let $f\in\calS^\prime(\RR^n)$ and $g\in\calS(\RR^n)$. We define the convolution between them as
\begin{equation}
\label{convSp}
\left(f\ast g,\varphi\right) := \left(f,\psi\ast\varphi\right), \qquad \forall\varphi\in\calS(\RR^n)
\end{equation}
where $\psi(x):=g(-x)$ and we have used the fact that, if $\psi,\varphi\in\calS$, then $\psi\ast\varphi\in\calS$ \cite{Vladimirov_GF}. 
\end{definition}

\noindent
Let us verify that this definition is proper, that is to say, the right-hand side of (\ref{convSp}) defines a continuous linear functional on $\calS^\prime$.
\begin{proposition}
\label{ConvSprime}
Let $f\in\calS^\prime(\RR^n)$ and $g\in\calS(\RR^n)$. The convolution $f\ast g\in\calS^\prime(\RR^n)$.
\end{proposition}

\begin{proof}
Let $\varphi,\phi\in\calS(\RR^n)$ and $a,b\in\RR$. The linearity is clear because 
\begin{equation}
(f,\psi\ast(a\varphi + b\phi)) = a(f,\psi\ast\varphi)+b(f,\psi\ast\phi) = a(f\ast g,\varphi)+b(f\ast g,\phi) 
\end{equation}
where we have used that the convolution between functions is linear and the fact that $f$ is a distribution. For continuity, first of all, we invoke the Schwartz's theorem \cite{Vladimirov},
\begin{align}
|\left(f,\psi\ast\varphi\right)| &\leq C ||\psi\ast\varphi||_s = C  \underset{x\in\mathbb{R}^n; |\alpha|\leq s}{\mathrm{sup}} \left(1+ |x|^2 \right)^\frac{s}{2}\left|D^\alpha(\psi\ast\varphi)(x)\right| \nonumber\\
&\leq C \underset{x\in\mathbb{R}^n; |\alpha|\leq s}{\mathrm{sup}} \left(1+ |x|^s \right)^\frac{s}{2}\left|(D^\alpha\psi\ast\varphi)(x)\right|\nonumber\\
&\leq C \underset{x\in\mathbb{R}^n; |\alpha|\leq s}{\mathrm{sup}} \int_{\RR^n} \D y \left(1+ |x|^2 \right)^\frac{s}{2} \:|D^\alpha\psi(y)|| \varphi(x-y)|\nonumber\\
& = C \underset{z\in\mathbb{R}^n; |\alpha|\leq s}{\mathrm{sup}} |\varphi(z)| \int_{\RR^n} \D y \left(1+ |z+y|^2 \right)^\frac{s}{2} \:|D^\alpha\psi(y)|\nonumber\\
& \leq C_1 \underset{z\in\mathbb{R}^n; |\alpha|\leq s}{\mathrm{sup}} |\varphi(z)| \int_{\RR^n} \D y \left(1+ |z|^2+|y|^2 \right)^\frac{s}{2} \:|D^\alpha\psi(y)|\nonumber\\
& \leq C_1 \underset{z\in\mathbb{R}^n; |\alpha|\leq s}{\mathrm{sup}} |\varphi(z)| \int_{\RR^n} \D y \left[\left(1+ |z|^2\right)^{^\frac{s}{2}}+|y|^s \right] \:|D^\alpha\psi(y)|\nonumber\\
&=  C_1 \underset{z\in\mathbb{R}^n; |\alpha|\leq s}{\mathrm{sup}} \left(1+ |z|^2\right)^{^\frac{s}{2}} |\varphi(z)| \int_{\RR^n} \D y |D^\alpha\psi(y)| \nonumber\\
&+ C_1 \underset{z\in\mathbb{R}^n; |\alpha|\leq s}{\mathrm{sup}} |\varphi(z)| \int_{\RR^n} \D y \:|y|^s |D^\alpha\psi(y)|\nonumber\\
&= C_2 \underset{z\in\mathbb{R}^n; |\alpha|\leq s}{\mathrm{sup}} \left(1+ |z|^2\right)^{^\frac{s}{2}} |\varphi(z)| + C_3 \underset{z\in\mathbb{R}^n; |\alpha|\leq s}{\mathrm{sup}} |\varphi(z)| \label{proofCSp} 
\end{align}
where $C_i\in\RR^+$ and, in the last step, we have used proposition (\ref{SsL}). Now, let us prove continuity. Let $\varphi_k$ be a sequence belonging to $\mathcal{S}$ that converges to $0$ as $k\rightarrow\infty$; therefore, having in mind equation (\ref{proofCSp}), $(f, \psi\ast\varphi_k) \rightarrow 0, \:\:\forall f$.
\end{proof}

\begin{proposition}
\label{StoS}
The Fourier transform and the inverse are continuous and linear operations from $\mathcal{S}^\prime$ to $\mathcal{S}^\prime$.
\end{proposition}

\begin{proof}
Let $g\in\mathcal{S}^\prime(\mathbb{R}^n)$, $\varphi,\phi\in\mathcal{S}(\mathbb{R}^n)$ and $a,b\in\RR$. The linearity is evident since
\begin{equation}
(\calF[g],a\phi+b\varphi) = (g, \calF[a\phi+b\varphi]) = a(g,\calF[\phi]) + b(g,\calF[\varphi]) = a(\calF[g],\phi) + b(\calF[g],\varphi)
\end{equation}
where we have use the Fourier transform in $\calS$ is linear and the fact that $g$ is a distribution. Let us prove now continuity. Let $g$ a continuous tempered distribution in the sense that it exists a sequence $g_k$ belonging to $\mathcal{S}^\prime$ that converges to $g$ as $k\rightarrow\infty$. Therefore, 
\begin{equation}
\lim_{k\rightarrow\infty}(\mathcal{F}[g_k],\varphi) = \lim_{k\rightarrow\infty}(g_k,\mathcal{F}[\varphi])= (g,\mathcal{F}[\varphi]) = (\mathcal{F}[g],\varphi),
\end{equation}
namely 
\begin{equation}
\lim_{k\rightarrow\infty} \mathcal{F}[g_k]  = \mathcal{F}[g] \;.
\end{equation}
For the inverse, the proof is equivalent. 
\end{proof}

\begin{proposition}
\label{FourierConv}
Let $f\in\calS^\prime(\RR^n)$ and $g\in\calS(\RR^n)$. Then,
\begin{equation}
\calF\left[f\ast g\right] = \calF\left[f\right]\cdot\calF\left[g\right] \in\calS^\prime \:.
\end{equation}
\end{proposition}

\begin{proof}
 Let $\varphi\in\calS(\RR^n)$. According to proposition (\ref{ConvSprime}), $f\ast g\in\calS^\prime$.  Moreover, because of proposition (\ref{StoS}), the Fourier transform of the convolution belongs to $\calS^\prime(\RR^n)$. Therefore, using (\ref{FTS}) and (\ref{PropConvSp}), we get
 \begin{equation}
 \label{convF1}
 \left(\calF\left[f\ast g\right],\varphi\right) =  \left(f\ast g,\calF\left[\varphi\right]\right) = \left(f,\psi\ast\calF\left[\varphi\right]\right)
\end{equation}
where $\psi(x):=g(-x)$. Using the fact that $\calF^{-1}\left[\calF\left[f\right]\right]= f \in \calS^\prime$, equation (\ref{convF1}) becomes
\begin{equation}
\left(f,\psi\ast\calF\left[\varphi\right]\right) = \left(\calF\left[f\right],\calF^{-1}\left[\psi\ast\calF\left[\varphi\right]\right]\right) =  \left(\calF\left[f\right],(2\pi)^n \calF^{-1}\left[\psi\right]\varphi\right)
\end{equation}
where we have used the property
\begin{equation}
\calF^{-1}\left[\phi\ast \varphi\right] = (2\pi)^n \calF^{-1}\left[\phi\right] \calF^{-1}\left[\varphi\right] \qquad \varphi,\phi\in\calS(\RR^n)\:.
\end{equation}
Finally, having in mind that
\begin{equation}
\calF^{-1}\left[\psi\right] = \frac{1}{(2\pi)^n} \calF\left[\psi(-x)\right]\:,
\end{equation}
we get
\begin{equation}
 \left(\calF\left[f\ast g\right],\varphi\right) = \left(\calF\left[f\right],\calF\left[g\right]\:\varphi\right) =  \left(\calF\left[f\right] \calF\left[g\right], \varphi\right)
\end{equation}
where the last step we have used proposition (\ref{SsSp}). Equivalently,
\begin{equation}
\calF\left[f\ast g\right] = \calF\left[f\right] \calF\left[g\right] \qquad \forall \varphi\:.
\end{equation}
\end{proof}

\bibliographystyle{utphys}
{\footnotesize\bibliography{References}}

\end{document}